\documentclass{article}
\usepackage{ijcai20}
 
\usepackage{times}
\usepackage{soul}
\usepackage{url}
\usepackage{pifont}

\usepackage{microtype}
\usepackage[colorinlistoftodos,obeyFinal]{todonotes}

\usepackage[hidelinks,breaklinks]{hyperref}

\pgfdeclarelayer{edgelayer}
\pgfdeclarelayer{nodelayer}
\pgfsetlayers{edgelayer,nodelayer,main}

\tikzstyle{none}=[inner sep=0pt]

\tikzstyle{rn}=[circle,fill=Red,draw=Black,line width=0.8 pt]
\tikzstyle{gn}=[circle,fill=White,draw=Black,line width=0.8 pt]
\tikzstyle{yn}=[circle,fill=Yellow,draw=Black,line width=0.8 pt]
\tikzstyle{simple}=[circle,fill=White,draw=Black]
\tikzstyle{newstyle1}=[circle,fill=Black,draw=Black,line width=0.3 pt,inner sep=0pt]
\tikzstyle{simple2}=[-,dashed,draw=Black]
\tikzstyle{simpledotted}=[-,dotted,draw=Black]
\tikzstyle{simple}=[-,draw=Black,line width=2.000]
\tikzstyle{arrow}=[-,draw=Black,postaction={decorate},decoration={markings,mark=at position .5 with {\arrow{>}}},line width=2.000]
\tikzstyle{tick}=[-,draw=Black,postaction={decorate},decoration={markings,mark=at position .5 with {\draw (0,-0.1) -- (0,0.1);}},line width=2.000]
\tikzstyle{newstyle2}=[-latex,draw=Black]
\tikzstyle{newstyle3}=[->,dotted,draw=Black]
\tikzstyle{newstyle6}=[->,dotted,draw=Black]

\usepackage{thm-restate}

\newcommand{\mypar}[1]{\medskip\noindent\textbf{#1.}}

\def\rConstrHelper(#1,#2,#3,#4){{#1}\ \substack{#2 \\ #3}\ {#4}}

\newcommand{\sss}{{\mathrel{\kern.25em{\sqsubseteq}\kern-.5em \mbox{{\scriptsize *}}\kern.25em}}}
\newcommand{\smallsss}{{\mathrel{\kern.25em{\sqsubseteq}\kern-.4em \mbox{{\scriptsize *}}\kern.25em}}}

\newcommand{\I}{\ensuremath{\mathcal{I}}\xspace}

\newcommand{\node}{\ensuremath{\mathop{\mathsf{node}}}}

\newcommand{\EL}{\ensuremath{{\cal E\!L}}\xspace}

\newlength{\indxlength}
\newcommand{\PTime}{\textsc{PTime}\xspace}

\newcommand{\PTimeL}{\textsc{PTimeL}\xspace}
\newcommand{\PTimePL}{\textsc{PTimePL}\xspace}

\newcommand{\Emc}{\ensuremath{\mathcal{E}}\xspace}

\newcommand{\Imc}{\ensuremath{\mathcal{I}}\xspace}

\newcommand{\Kmc}{\ensuremath{\mathcal{K}}\xspace}
\newcommand{\Lmc}{\ensuremath{\mathcal{L}}\xspace}

\newcommand{\Fmf}{\ensuremath{\mathfrak{F}}\xspace}

\newcommand{\Rsf}{\ensuremath{\mathsf{R}}\xspace}

\usepackage[utf8]{inputenc}
\usepackage[small]{caption}
\usepackage{graphicx}
\usepackage{amsmath}
\usepackage{amsthm}
\usepackage{amsmath}
\usepackage{amsfonts}
\usepackage{amssymb}
\usepackage{graphicx}
\usepackage{mathtools}
\usepackage{tikz}
\usepackage{booktabs}
\usepackage[ruled]{algorithm2e}
 \newcommand{\examples}{\ensuremath{\Emc}\xspace}

\newcommand{\target}{\ensuremath{t}\xspace} 
\newcommand{\hypothesisSpace}{\ensuremath{\Lmc}\xspace}

\newcommand{\e}{\ensuremath{e}\xspace}

\newtheorem{corollary}{Corollary}
\newtheorem{theorem}{Theorem}

\newtheorem{example}[theorem]{Example}
\newtheorem{remark}{Remark}
\newtheorem{lemma}[theorem]{Lemma}

\newtheorem{claim}{Claim}[theorem]

\newcommand{\MQc}[2]{\ensuremath{{\sf MQ}_{#1,#2}\xspace}}
\newcommand{\EQc}[2]{\ensuremath{{\sf EQ}_{#1,#2}\xspace}}

\newcommand{\val}{\ensuremath{{\sf val}}\xspace}
\newcommand{\inc}{\ensuremath{{\sf inc}}\xspace}

\newcommand{\p}{\ensuremath{{\sf prec}}\xspace}
\newcommand{\safe}{safe\xspace}
\newsavebox{\spacebox}
\begin{lrbox}{\spacebox} 
\verb*! !
\end{lrbox}

\usepackage[absolute,overlay]{textpos}
  \newcommand{\N}{\ensuremath{N}\xspace}
 \newcommand{\formula}{\ensuremath{\varphi}\xspace}
 \newcommand{\pos}{\ensuremath{\alpha}\xspace}
 
 \newcommand{\kb}{\emph{\ensuremath{\Kmc}}\xspace}
 \newcommand{\alphabar}{\ensuremath{\overline{\alpha}}\xspace}
 \newcommand{\algfindval}{\ensuremath{\mathsf{FindValuation}_t}\xspace}

\title{On the Learnability of  Possibilistic Theories}

\author{Cosimo Persia \And Ana Ozaki\footnote{Contact Author: Ana.Ozaki@uib.no
\newline
\copyright The sole copyright holder is IJCAI (\href{www.ijcai.org}{International Joint Conferences on Artificial Intelligence}), all rights reserved.
}
\affiliations
University of Bergen  
}

\begin{document}

%\copyrightstatement

\maketitle

\begin{abstract}
  We investigate learnability 
  of possibilistic  theories 
   from entailments in light of Angluin's exact learning model.
          We consider cases in which only membership,
  only equivalence, and both kinds of  queries can be posed by the learner.
  We then show that, for a large class of problems,
    polynomial time learnability results for classical logic 
  can be transferred
      to the respective possibilistic extension. 
    In particular, it follows from our results that the possibilistic 
  extension of propositional Horn theories is 
  exactly learnable in polynomial time.
  As polynomial time learnability   in the exact model 
  is transferable to the   classical  probably 
  approximately correct   model extended with membership queries, 
  our work also establishes such results
  in this model.                   \end{abstract}

\section{Introduction}
 
 Uncertainty is found in 
many phases of learning,
such as model selection and
processing noisy, imperfect, incomplete or limited data.
In most cases,
knowledge-based
systems are constrained to live under conditions of ignorance.
There are different approaches to deal with uncertainty \cite{Parsons1998}.
A well-studied formalism for dealing 
with it is 
\emph{possibilistic logic} \cite{duboisPosLog,Lang2000}.
It admits a graded notion of possibility
and
 makes a clear distinction between 
the concepts of truth and belief \cite{possvsmultivalued}.
Uncertainty of formulas in possibilistic logic is not subject to the 
complement rule as in probability theory
\cite{Parul:2015:2248-9622:37,327367}.
Indeed,
complementary  formulas may be considered fully possible,  
meaning complete ignorance
about
their truth value.

\begin{example}\label{ex:intro}\upshape
		Consider a doctor who  
	has to  diagnose
	 a patient that suffers from extreme fatigue.
	A doctor can consider 
	blood-related conditions: iron deficiency, iron overload, and vitamin B12 deficiency. 	Within possibility theory, one can model cases of complete uncertainty.
			Both 
	iron deficiency and iron overload, which are two mutually exclusive conditions,
	can be considered  
	fully possible.
	Consider that vitamin B12 deficiency is considered to be less possible, e.g. associated with the value $1/3$,
	based on some information provided by the patient.
	In probability theory, complete ignorance of the first two conditions
	  would make us assign probability $1/3$ to every condition (Laplace criterion).
	  Thus, it would not model the
	  knowledge about vitamin B12 deficiency and the 	  ignorance about
	  iron deficiency and iron overload. \hfill {\mbox{$\triangleleft$}}
	  
									\end{example}

Although  possibilistic logic  
has been extensively studied \cite{DBLP:reference/sp/DuboisP15}, 
there are not many works that investigate learnability of 
possibilistic theories. 
In this work, we partially cover this gap by
studying whether 
possibilistic theories are learnable in   Angluin's 
exact learning model~\cite{angluinqueries}.
In this model, a learner 
interacts with a teacher to exactly identify an abstract target concept. One can see the doctor, in Example~\ref{ex:intro}, as a learner who inquires 
%attempts 
%to acquire information from
%poses queries
the patient (playing the role of a teacher) in order to  identify a disease. %by posing queries.

The most studied communication protocol in this model 
contains  queries of two kinds, called \emph{membership} and \emph{equivalence} queries. 
Membership queries allow the learner to know whether
a certain statement holds. % (e.g. ``Anaemia in family history?'').
Equivalence queries allow the learner to check whether a hypothesis (e.g. a  diagnose)
is correct and, if not, to fix it using a counterexample.
In our toy scenario, the patient may not be able to provide a counterexample
but new symptoms or reactions can reveal that the hypothesis is not correct.  
To the best of our knowledge, this is the first work where 
learnability of possibilistic theories is investigated in Angluin's model. 
We consider cases in which only membership,
only equivalence, and both kinds of  queries can be posed by the learner.
We also study whether known   polynomial time exact learning results 
for classical logic can be transferred to possibilistic settings.

Our main result is that, for a large class of problems, 
polynomial time learnability  (with both types of queries)
 can be transferred
from classical logic 
      to the respective possibilistic extension (Theorem~\ref{thm:mqeq}). 
If only membership queries are allowed 
(and the maximal precision 
of valuations in the target is  fixed) then 
polynomial time learnability of a classical logic can also be transferred to the
possibilistic extension.
We leave  open the case in which only equivalence queries can be asked.
With our main result, 
we establish, e.g., that the possibilistic extension 
of  propositional Horn~\cite{DBLP:journals/ml/AngluinFP92,DBLP:conf/icml/FrazierP93,DBLP:journals/toct/HermoO20} and fragments of first-order Horn~\cite{arimura1997learning,reddy1998learning,KLOW18} are exactly
learnable in polynomial time. % (with both kinds of queries).
As polynomial time learnability   in the exact model 
  is transferable to the       probably 
  approximately correct (PAC)~\cite{Valiant}    model extended with membership queries, 
  our work also establishes such results
  in this model.

\mypar{Related Work}
Among the works that
combine learning and possibilistic logic,
we can find results 
on 
learning 
possibilistic logic theories from default rules
within the PAC learning model~\cite{DBLP:conf/ijcai/KuzelkaDS16}.
Possibilistic logic has been used 
to reason with default rules~\cite{DBLP:conf/kr/BenferhatDP92}
to
select the most plausible rule 
and
in inductive logic programming 
to handle exceptions~\cite{DBLP:journals/ai/SerrurierP07}.
In statistical relational learning,
possibilistic logic has been used as
a formal encoding of statistical regularities
found in relational data~\cite{DBLP:conf/ijcai/KuzelkaDS17}.
Possibilistic formulas can encode 
Markov logic networks~\cite{DBLP:conf/uai/KuzelkaDS15}.
Formal concept analysis 
has been applied to generate 
attribute implications with a degree of certainty~\cite{Djouadi10}.
We also point out an extension of version space learning
that deals with examples associated with possibility degrees~\cite{doi:10.1002/int.20310}.

\smallskip

In Section~\ref{sec:basic}, we present basic definitions.
In Section~\ref{sec:learnability}, we investigate
whether possibilistic logic theories can be learned and,
in Section~\ref{sec:reduction}, we show transferability of polynomial time learnability
results.

\section{Basics}
\label{sec:basic}
In the following, we provide relevant notions of possibilistic logic and learning theory 
used in  the paper. 
\subsection{Possibilistic Theories}

Let $L$ be a propositional or a first-order  (FO) language (restricted to 
well-formed formulas without free variables) with the semantics of classical FO logic.
We say that $\formula\in L$ is \emph{satisfiable} if there is 
an interpretation $\Imc$ such that $\formula$ is satisfied
in $\Imc$. Moreover, $\formula$ is \emph{falsifiable} if its negation $\lnot \formula$ is satisfiable.
An \emph{FO knowledge base} (FO KB) is a finite set of FO formulas. 
An FO KB is \emph{non-trivial} 
if it is satisfiable and falsifiable.
The \emph{possibilistic extension} of an FO language $L$ is defined as follows. A \emph{possibilistic formula} is a pair $(\formula,\pos)$,
where
$\formula\in L$
and
$\pos$ is a real  number (with finite precision) in the interval $(0,1]$, called the \emph{valuation} of $\formula$.
A \emph{possibilistic KB} (or a possibilistic theory) is a 
finite 
set
$\kb$
of possibilistic formulas.   
Given a set $\Omega$ of
interpretations for $L$,
a \emph{possibility distribution} $\pi$ is a function from $\Omega$ to the interval $[0,1]$. The \emph{possibility}  
and \emph{necessity measures}, $\Pi$ and $\N$, are functions (induced by $\pi$) from $L$ to $[0, 1]$, defined respectively as
\[\Pi(\varphi)=\sup\{\pi(\Imc)\mid \Imc\in\Omega, \ \Imc\models\varphi\}\] \[  N(\varphi) = 1-\Pi(\neg\varphi)
 = \inf\{ 1- \pi(\I) \mid \Imc\in\Omega, \ \I \models \neg \varphi \}.\] 
   A possibility distribution $\pi$ \emph{satisfies}  a 
 possibilistic formula $(\formula,\pos)$,
 written $\pi \models (\varphi, \pos)$,
 if $\N(\varphi)\geq \pos$,
 and it satisfies a possibilistic KB
 $\kb = \{ (\formula_{i},\alpha_{i}) \mid 0 \leq i < n \}$
 if it satisfies each
$(\formula_{i},\alpha_{i}) \in \kb$.
We have that  $(\formula,\pos)$ is \emph{entailed} by $\kb$, 
written $\kb\models (\formula,\pos)$, if all possibility distributions 
that satisfy \kb also satisfy $(\formula,\pos)$.
Given $\kb$ as above and $\Imc \in \Omega$, we define the possibility distribution $\pi_\kb$ as follows:
$\pi_\kb(\Imc) = 1$, if $\Imc \models \formula_{i}$, for every $ (\formula_{i},\alpha_{i}) \in \kb$; otherwise,
$\pi_\kb(\Imc) = \min \{ 1 - \alpha_{i} \mid \Imc \models \lnot \formula_{i}, 0 \leq i < n \}$.

The \emph{FO projection} of $\kb$ is the set
$\kb^{\ast} = \{ \formula_{i} \mid (\formula_{i}, \alpha_{i}) \in \kb \}.$
The \emph{$\alpha$-cut} and the \emph{$\alphabar$-cut} of $\kb$,
with $\alpha \in (0,1]$,
are defined respectively as
$\kb_\alpha =  \{ (\formula,\beta) \in \kb \mid \beta \geq \alpha \}$
and
$\kb_{\alphabar} = \{ (\formula, \beta) \in \kb \mid \beta > \alpha \}$.
The set of all valuations occurring in $\kb$ is 
$\kb^{v} = \{ \alpha \mid (\formula, \alpha) \in \kb \}$.
Moreover, 
$\val(\formula, \kb)=\sup\{ \alpha \mid \kb \models(\formula,\alpha) \}$ is the least upper bound of the valuations of formulas entailed by $\kb$.
Finally, the \emph{inconsistency degree} of $\kb$ is defined as
$\inc(\kb) = \sup\{ \alpha \mid \kb \models (\bot, \alpha) \}$.

\begin{lemma}\cite{duboisPosLog} 		\label{p:propositions}
	Let \Kmc be a possibilistic KB. For every possibilistic formula 
	$(\phi,\alpha)$,  	\begin{enumerate}
	\item $\Kmc \models (\phi,\alpha)$ 
	iff $\Kmc_\alpha^\ast \models \phi $; 		\item $\Kmc \models (\phi,\alpha)$ iff 
	 $\alpha\leq \val(\phi,\Kmc)$; and 
	 \item  $ \Kmc \models (\phi,\alpha) $ implies $ \val(\phi,\Kmc)\in \Kmc^v \cup\{1\}$.
	  	\end{enumerate}
\end{lemma}
\begin{proof}
Point~1 is a consequence of 
Propositions 3.5.2, 3.5.5, and 3.5.6,  and Point~2 is Property~1 at page 453 in \cite{duboisPosLog}.
We argue about Point~3.
By definition of $\pi_\Kmc$, for all $\Imc\in\Omega$,
		 $\pi_\Kmc(\I) $ is either $ 1 $ or $ 1-\beta $ for some $ \beta\in\Kmc^v $.
		 Let $N_\Kmc$ be the necessity measure induced by $\pi_\Kmc$.
		 		By definition of $N_\Kmc$, $ N_\Kmc(\phi)  = \inf\{ 1- \pi_\Kmc(\I) \mid \Imc\in\Omega, \ \I \models \neg \varphi \}$. 
Then, $ N_\Kmc(\phi) \in  \Kmc^v \cup\{0,1\}$ (recall that $\inf\{\}$ is $1$, which is the case for tautologies). 
		By the semantics of possibilistic logic, $ N_\Kmc(\phi)=\val(\phi,\Kmc)$~\cite[Corollary~3.2.3]{duboisPosLog}. As $(\phi,\alpha)$ is a possibilistic formula, $\alpha>0$. So, by Point~2,  $ N_\Kmc(\phi)=\val(\phi,\Kmc)\in \Kmc^v \cup\{1\}$. 
\end{proof}

We denote by $ =_p $ the operator that checks if two numbers are equal up to precision $ p $.  
For example $ 0.124 =_2 0.12345 $ but $ 0.124 \neq_3 0.12345 $.
Assume $ \alpha\in (0,1] $  has finite precision.
We write $\p(\alpha)$ for the precision of $ \alpha $ and
$\p(t)$ for $sup\{\p(\alpha) \mid (\phi,\alpha)\in t\}$.
Given an interval $I$, we write $ I_p $ for the set containing all $ \alpha \in I$
with $ \p(\alpha) = p$.

\begin{example}\label{ex:iron}\upshape
One can express (1) mutual exclusion of iron deficiency and iron overload
	and (2)
	lower necessity of iron overload to be the cause of fatigue
	than iron deficiency with 
	the possibilistic KB $ \{(\forall x({\sf IronDef}(x) \rightarrow \neg {\sf IronOver}(x)),1),
	(\forall x({\sf IronDef}(x) \rightarrow {\sf Fatigue}(x)),0.9),
	(\forall x({\sf IronOver}(x) \rightarrow {\sf Fatigue}(x),0.8)\} $. \hfill {\mbox{$\triangleleft$}}
	\end{example}	
\subsection{Learnability}\label{subsec:learnability}

In learning theory, examples are pieces of
information that characterise an abstract target the learner wants to learn.  
We consider the problem of learning targets represented in 
decidable fragments of 
FO logic 
or in their  possibilistic extensions.
Examples in our case are
formulas expressed in the chosen logic (in this context called `entailments'). 

A \emph{learning framework} $\Fmf$ is a pair 
$(\examples, \hypothesisSpace)$; where $\examples$ 
is a non-empty  and countable set of \emph{examples}, and  $\hypothesisSpace$ is a non-empty and countable set of \emph{ concept representations} (also called \emph{hypothesis space}). 
Each element $l$ of $\hypothesisSpace$ is assumed to be represented using 
a finite set of symbols $\Sigma_l$ (the \emph{signature} of $l$). In all learning frameworks considered in this work, $\examples$ is a set of formulas and $\hypothesisSpace$ 
is a set
of KBs
(in a chosen language).
We say that $\e\in \examples$ is a \emph{positive example} for $l \in \hypothesisSpace$ if $l\models\e$ and a
\emph{negative example} for $l$ if $l\not\models\e$. 
Given a learning framework $\Fmf = (\examples, \hypothesisSpace)$, we are interested in the
exact identification of a \emph{target}   $\target\in\hypothesisSpace$,
by posing queries to oracles.
Let  ${\sf MQ}_{\Fmf,\target}$ be the oracle that takes as input some $\e \in \examples$ and
returns `yes' if $\target\models\e$ and `no' otherwise. 
A \emph{membership query} is a call to the oracle ${\sf MQ}_{\Fmf,\target}$.
Given $t,h\in\hypothesisSpace$,
a \emph{counterexample} for $t$ and $h$ is an example $\e\in\examples$ s.t. 
$t\models \e$ and $h\not\models \e$ (or vice-versa, $h\models \e$ and $t\not\models \e$).
For every $\target \in \hypothesisSpace$, we denote by ${\sf EQ}_{\Fmf,\target}$ an oracle
that takes as input a \emph{hypothesis} $h \in \hypothesisSpace$
and returns `yes' if $h\equiv\target$ and a counterexample 
otherwise.
There is no assumption regarding which counterexample  
is chosen by the oracle.  
An \emph{equivalence query} is a call to  ${\sf EQ}_{\Fmf,\target}$.

\begin{example}\upshape		A blood test to check for vitamin B12 deficiency
	on patient $42$ can
	 	 			be modelled with a call
	to ${\sf MQ}_{\Fmf,\target}$ with 	$ ({\sf B12Def}({\sf patient\_42}),\alpha)$ for some $\alpha\in (0,1]$
		as input (depending on the result and accuracy of the test). \hfill {\mbox{$\triangleleft$}}
\end{example}

A \emph{learner} 
for $\Fmf=(\examples,\hypothesisSpace)$ is a deterministic  algorithm that, for a fixed but arbitrary $\target\in\hypothesisSpace$, 
takes $\Sigma_t$  as input, is allowed to pose queries to ${\sf MQ}_{\Fmf,\target}$ and
${\sf EQ}_{\Fmf,\target}$ (without knowing  the target $\target$), and that eventually halts and outputs some $h\in\hypothesisSpace$ with
$h\equiv\target$.
This notion of an algorithm   with access to oracles  can be formalised using 
\emph{learning systems}~\cite{watanabe90}, where posing a query to 
an oracle means writing down 
the query in an (additional) communication tape, entering in a query state, and waiting.
The oracle then
writes the answer in the communication tape, enters in an answer state, and stops.
After that, the learner resumes its execution and can now read the answer in the communication tape.

We say that \Fmf is (exactly) \emph{learnable} if there is a learner   for \Fmf and that \Fmf is 
  \emph{polynomial time learnable} if it is learnable by a
learner $A$ such that at every step (the time used by an oracle to write an answer is \emph{not} taken into account) 
of computation the time used by $A$  up to
that step is bounded by a polynomial $p(|\target|,|\e|)$, where $\target\in \hypothesisSpace$
is the target and $\e \in \examples$ is the largest counterexample seen so far.
We denote by \PTimeL 
the class of learning frameworks
which are
 polynomial time 
 learnable 
  and the complexity of the entailment problem is in \PTime
 \footnote{In general, non-trivial algorithms
 need to perform entailment checks to combine
 the information of the examples.
 So polynomial time learning algorithms are normally for logics
in which the entailment problem is tractable.
This is the case e.g. for the Horn results
mentioned in the Introduction.}.
 We also consider cases in which the learner can only 
pose one type of query (only membership or only equivalence queries).
Whenever this is the case we write this explicitly.

	Let $\Fmf=(\examples,\hypothesisSpace)$ be a  learning framework
	where
	$\examples$ is a set of FO formulas and
	$\hypothesisSpace$ is a set of
	FO KBs. 
	We call such \Fmf  an \emph{FO learning framework}.
	We say that $\Fmf$ is \emph{non-trivial} 
	if $\hypothesisSpace$ contains a non-trivial FO KB; and 
	that it is \emph{\safe} if  $l\in\hypothesisSpace$ implies that
	 $l'\in \hypothesisSpace $, for all $l'\subseteq l$.
	 		A \emph{possibilistic extension} $l_\pi$ of an FO 
KB
$l$ is  a possibilistic KB obtained by adding a possibilistic valuation $ \alpha  $
	to every formula  $ \varphi \in l$. 
		The possibilistic extension $\Fmf_\pi$ of \Fmf 
is the pair $(\examples_\pi,\hypothesisSpace_\pi)$ where 
$\hypothesisSpace_\pi$
is the set of all possibilistic extensions of each $l\in \hypothesisSpace$, 
and $ \examples_\pi $ is the set of all possibilistic formulas entailed by an element of $\hypothesisSpace_\pi$.

We write $\mathbb{N}^+$ for the set of positive natural numbers.
Given 
$p\in \mathbb{N}^+$, we denote by  $\Fmf_\pi^p =(\examples_\pi,\hypothesisSpace^p_\pi)$  the result of removing from  $\hypothesisSpace_\pi$ in $\Fmf_\pi$
	 every $l\in \hypothesisSpace_\pi $ that does not satisfy $ \p(l) = p $.

\begin{remark}\label{remark:membership}	Let $\Fmf=(\examples,\hypothesisSpace)$ be an FO learning framework and let $t\in\hypothesisSpace$ be the target. 
If a learner $A$ has access to ${\sf MQ}_{\Fmf,\target}$
	then 	we can  assume w.l.o.g. that all counterexamples returned by  ${\sf EQ}_{\Fmf,\target}$ are \emph{positive}:
	the learner can  check whether each  $\phi\in h$
	is entailed by $t$. The same holds for $\Fmf_\pi$.
\end{remark}

\section{Learnability Results} \label{sec:learnability}
We start by studying the problem of whether 
there is a learner for a learning framework such that it always terminates with
a hypothesis equivalent to the target. The main difficulty in learning with only membership queries (even for plain FO settings)
is that  the learner would 
`not know' whether it has found a formula 
equivalent to a (non-trivial) target. 
\begin{example}\label{e:mq:nontrivialconsequence}\upshape 	Let $\Phi_n := \exists  x_1 \ldots \exists  x_n .\bigwedge_{0\leq i < n}r(x_i,x_{i+1})$.
													A learner may ask membership queries  of the form  $\exists x_0\Phi_n$ 	for an arbitrarily large $n$  without being able to distinguish whether  the target 
	theory is $\exists x_0\Phi_n$ or $\forall x_0(\Phi_n\rightarrow \Phi_{n+1})$ (knowing the signature of the target theory 
	does not help the learner). \hfill {\mbox{$\triangleleft$}} 
\end{example}

For possibilistic theories, another difficulty arises even for the propositional case. As
the precision of a formula can be
arbitrarily high, the learner may not know when to stop (e.g.,
is the target $(p,0.1)$? or $(p,0.11)$?). 
Theorem~\ref{thm:memonly} states that, except for trivial cases, learnability cannot be guaranteed. 

\begin{restatable}{theorem}{memonly}\label{thm:memonly}
				Let \Fmf be a non-trivial FO learning framework.
	$\Fmf_\pi$ 
	is not (exactly) 
	learnable with only
	membership queries.
\end{restatable}
\begin{proof}[Sketch]
	The existence of a learner $A$ for the possibilistic extension $\Fmf_\pi=(\examples_\pi,\hypothesisSpace_\pi)$ of a non-trivial learning framework \Fmf
	would imply the existence of a procedure that 
	terminates in $n$ steps. $A$ 	would not distinguish between 
	the elements of $\hypothesisSpace_\pi$ with precision higher than $n$. 
\end{proof}

If the precision of the target is known or fixed, 
  learnability of an FO learning framework 
can be transferred to its possibilistic extension.
We state this in Theorem~\ref{thm:memonlyprec}.
To show this theorem, we use the following 
technical result.

				\begin{lemma}
	\label{l:fundamentalidea}
				Let $t$ be a possibilistic KB. 
	Let $I$ be a set of valuations such that $t^v\subseteq I$. 
	 If for each
	$ \alpha\in I$ there is some FO KB $k_\alpha^\ast$ such that 
	$ k_\alpha^\ast \equiv t_\alpha^\ast$
	then  
	$ t \equiv \{ (\phi,\alpha) \mid 
	\phi \in k_\alpha^\ast, \alpha\in I\}  $. 
		\end{lemma}
\begin{proof}
Let $h=\{ (\phi,\alpha) \mid 
	\phi \in k_\alpha^\ast, \alpha\in I\}$. 	Assume $ h \models (\phi,\gamma) $. 
	If $\gamma=1$ and $\gamma\not\in I$ then 	$\phi$ is a tautology. In this case, for all $\beta\in (0,1]$, $ t \models (\phi,\beta) $.
	Suppose this is not the case.
	By Points~2 and~3 of Lemma~\ref{p:propositions},
	$ \gamma \leq \alpha$,
	 $\alpha = \val(\phi,h)
	\in h^v \cup \{1\}$. Also, $ h \models (\phi,\alpha) $. 
	By construction of $ h $, $h^v=I$, so $\alpha\in I$.
							Moreover, 
	for every $ \beta\in I$, we know that $ h^*_\beta = k^*_\beta $.
	Therefore $ k^*_\alpha \equiv h^*_\alpha $.
	By Point~1 of Lemma~\ref{p:propositions},
	$ h \models (\phi,\alpha) $ implies 
	$ h_\alpha^* \models \phi $. 
	Then, 
			$ k_\alpha^\ast \models \phi $.
	As $ k_\alpha^\ast \equiv t_\alpha^\ast$, we have that 
		$ t_\alpha^\ast \models \phi $.
	Again by Point~1 (of Lemma~\ref{p:propositions}), 
	$ t_\alpha^\ast \models \phi $ iff $ t \models (\phi,\alpha) $.
	Since $\alpha\geq\gamma$, $ t \models (\phi,\gamma) $ by Point~2. 	The other direction can be proved similarly.
\end{proof}

\begin{restatable}{theorem}{memonlyprec}\label{thm:memonlyprec} 	Suppose \Fmf is an 
	FO learning framework that is learnable  with only membership queries.  
	For all $p\in\mathbb{N}^+$,  $\Fmf^p_\pi=(\examples_\pi,\hypothesisSpace^p_\pi)$ 	is   	learnable with only
	membership queries.
\end{restatable}
\begin{proof}
	Let $ A $ be a 
		learner
	for \Fmf and let $ t\in \hypothesisSpace_\pi^p $ be the target.
	For each $ \alpha\in (0,1]_p $, 	we run an instance of $ A $, denoted $ A_\alpha $. 		Whenever  $ A_\alpha$ calls 	$ \MQc{\Fmf}{t^*_\alpha} $ 
	with $ \phi $ as input,
	we call 
	$ \MQc{\Fmf_\pi}{t} $ 
	with $ (\phi,\alpha) $ as input.
	By Point~1 of Lemma~\ref{p:propositions},
	$ \MQc{\Fmf}{t^*_\alpha}(\phi) = \MQc{\Fmf_\pi}{t}(\phi,\alpha) $.
	Since $ A $ is a learner for \Fmf, every $ A_\alpha $ eventually halts and 
	outputs  a hypothesis 
	$ k_\alpha^* $ 	such that
	$ k_\alpha^* \equiv t^*_\alpha$.
	Since $ t\in \hypothesisSpace_\pi^p $, 
	$t^v\subseteq (0,1]_p$.
			By Lemma~\ref{l:fundamentalidea},
	$ t \equiv \{ (\phi,\alpha) \mid 
	\phi \in k_\alpha^\ast, \alpha\in (0,1]_p\}  $. 
			Thus, we can transfer learnability of \Fmf (with only membership queries)
	 to $ \Fmf_\pi^p$.
\end{proof}

If, e.g., $\MQc{\Fmf_\pi}{t}((\phi,0.01) )=$ `yes',
 $\MQc{\Fmf_\pi}{t}((\phi,0.02))=$ `no', and the precision of the target is $ 2 $,  then $ \val(\phi,t) = 0.01 $.
So, knowing the precision is important for 
learning with membership queries only. 
If equivalence queries are allowed then 
a learner can build a hypothesis equivalent to  
the target \emph{without knowing the precision} in advance by
simply enumerating all possible hypothesis and asking them
to the oracle, one by one (Theorem~\ref{thm:eqonlyq}).

\begin{restatable}{theorem}{memonly}\label{thm:eqonlyq}
		The possibilistic extension $ \Fmf_\pi $ of an FO learning framework \Fmf
	is   	learnable with only
	equivalence queries.
\end{restatable}
%~ \begin{proof} Every %FO learning framework
%~ \Fmf is learnable with only equivalence queries.
%~ Indeed, a naive learner $A$
%~ for \Fmf is one that  enumerates all $l\in\Lmc$ built using 
%~ symbols from $\Sigma_t$,  taken as input, and asks all the possible hypothesis 
%~ to the oracle ${\sf EQ}_{\Fmf,\target}$. %, one by one. 
%~ %The learner does not  know the size of the target in advance but it can 
%~ %estimate it to be $n$, ask all possible hypothesis of this size, then increase 
%~ %to $n+1$, and so on. 
%~ Since the target is finite, eventually $A$   halts and outputs 
%~ $h$ equivalent to $t$. 
%~ For $ \Fmf_\pi $, a similar  learner $A_\pi$ exists, but in this case, it also 
%~ needs to take into account the precision of the target.
%~ %to estimate the precision of the target and increase it as it navigates the 
%~ %search space.
%~ %As the precision of the target is finite, eventually $A_\pi$ also     halts and outputs 
%~ %an equivalent hypothesis.
%~ \end{proof}

If both membership and equivalence query oracles
are available,
learnability is guaranteed by the previous  theorem. 
%(even when the precision of the target is unknown). 

\begin{restatable}{corollary}{memonly}\label{thm:eqonly} 
Let \Fmf be an FO learning framework. 
	\Fmf is learnable iff $ \Fmf_\pi $ is learnable.
\end{restatable}

\section{Polynomial Time Reduction}\label{sec:reduction}
We now investigate whether 
 results showing that an FO learning framework is in \PTimeL can be transferred  
to
their possibilistic extensions and vice-versa.
Theorem~\ref{t:posstoclassical} shows
the transferability of 
\PTimeL membership from the possibilistic extension $ \Fmf_\pi $  of an FO learning framework \Fmf 
to \Fmf.

\begin{theorem}\label{t:posstoclassical}
Let \Fmf be an FO learning framework. 
	If $ \Fmf_\pi $ is in \PTimeL
	then $ \Fmf $ is in \PTimeL.
\end{theorem}
\begin{proof}
In our proof, we use the following claim.
\begin{claim}\label{l:target2poss}
		Let $k$ be an FO KB and let $ t$ be the possibilistic KB $ \{ (\phi,1)\mid \phi\in k \} $.
	For all $(\phi,\alpha)$, 	 $ k\models \phi $
	iff $ t\models (\phi,\alpha) $.
	\end{claim}
\begin{proof}
		If 	$ t \models(\phi,\alpha) $,
	since $ t^* \models t^*_\alpha $ and $ k = t^* $,
			$ k\models\phi $.
	If 	$ k\models \phi$,
	by construction $ t^*_1\models\phi $.
	By Point 1 of Lemma~\ref{p:propositions},
	$ t^*_1\models\phi $ iff  $ t\models(\phi,1) $, 
	so, for all $\alpha\in (0,1]$, $ t\models(\phi,\alpha) $.
\end{proof}

	Let $\Fmf=(\examples, \hypothesisSpace)$ and let $ k \in\hypothesisSpace$  be the target.
	Since $ \Fmf_\pi $ is in \PTimeL,
	there is a learner $ A_\pi $ for $ \Fmf_\pi $.
	We start the execution of $ A_\pi $ that attempts to learn a hypothesis $ h $ 	equivalent to $ t=\{ (\phi,1)\mid \phi\in k \} $.
		By Claim~\ref{l:target2poss}, for all $\alpha\in (0,1]$, 
		$ \MQc{\Fmf_\pi}{t} ( (\phi,\alpha)) =\MQc{\Fmf}{k} (\phi)$. 
				Also, 
	we can simulate a call to \EQc{\Fmf_\pi}{t} with  $ h $ as input
	by calling \EQc{\Fmf}{k} with  $ h^* $ as input.
	By Claim~\ref{l:target2poss},		for all $\alpha\in (0,1]$, $ k\models \phi $ iff $ t\models (\phi,\alpha) $, 
		in particular, for $\alpha=1$.
	            By Remark~\ref{remark:membership}, we can assume that all counterexamples returned by $\EQc{\Fmf}{k}$ are positive. 
	Whenever we receive a (positive) counterexample $ \phi $,
    we return $ (\phi,1) $ to $ A_\pi $. 
        									Eventually, $ A_\pi $ will output a hypothesis $ h\equiv t $ 
		in polynomial time w.r.t. $ |t| $ and the largest counterexample received so far.
	Clearly, $h^*$ is as required. 				\end{proof}

By 
Theorem~\ref{t:existspossnotlearnable}, the converse of Theorem~\ref{t:posstoclassical} does not hold.
%~ Simple FO learning frameworks can become difficult to 
%~ learn when extended with possibilistic valuations 
%~ because algorithms also have to deal with 
%~ multiple valuations. 

\begin{theorem}\label{t:existspossnotlearnable}
	There exists an FO learning framework
	$ \Fmf $ such that 
	$ \Fmf $ 
	is 
	in \PTimeL
	but $ \Fmf_\pi = (\examples_\pi,\hypothesisSpace_\pi) $ is not in \PTimeL.
\end{theorem}
\begin{proof}
	Let $ \Fmf=(\examples,\hypothesisSpace) $ be an FO learning framework that is \emph{not} in \PTimeL. Such \Fmf exists, 
	one can consider, for instance, 
			the \EL learning framework~\cite[Theorem~68]{KLOW18}\footnote{Non-polynomial query learnability is 
	proved in~\cite[Theorem~68]{KLOW18}, which 	implies non-polynomial time learnability.}.
	We use $ \Fmf $ to define the learning framework 
	$ \Fmf^\bot = (\examples,\hypothesisSpace^\bot) $ where
	$ \hypothesisSpace^\bot = \{ h\cup\{ \phi,\neg\phi  \} \mid h\in \hypothesisSpace \}$ for a fixed but arbitrary non-trivial FO formula $ \phi $. 
			Even though \Fmf is not learnable
		in polynomial time,
	$  \Fmf^\bot$ is.
		The learner can learn any $ l\in \hypothesisSpace^\bot $ 
				by   returning the hypothesis $ \{ \bot \} $ (in constant time).
			Assume that $ \Fmf^\bot_\pi =(\examples_\pi, \hypothesisSpace_\pi^\bot)$ 
	is in \PTimeL.
	This means that for every target $ l\in \hypothesisSpace_\pi^\bot $
	we can learn 
		in polynomial time 
	a hypothesis $h$ such that $ h\equiv l $.
	By construction, for every $ t\in\hypothesisSpace $ there is  
	$ l\in\hypothesisSpace_\pi^\bot $ such that $ t\equiv l^{*}_{\overline{\inc(l)}} $.
	By learning $h$ such that 
	$ h\equiv l$
	we have also learned a hypothesis 
	$ h $
	such that $ h^*_{\overline{\inc(h)}} \equiv t $.
			By Theorem~\ref{t:posstoclassical}, 	$ \Fmf \in \PTimeL$, which contradicts our assumption that this is not the case.
	Therefore we have found an FO  
	learning framework $ \Fmf^\bot $
	that  is
		is in \PTimeL but 
	its 
	possibilistic extension $ \Fmf_\pi^\bot $ is not
		 in
		\PTimeL.
	\end{proof}

The FO learning framework $\Fmf^\bot$ in
the proof of Theorem~\ref{t:existspossnotlearnable} 
is not \safe (see definition in Subsection~\ref{subsec:learnability})
because, for $l\not\subseteq \{\phi,\neg \phi\}$ we have $l\in\hypothesisSpace^\bot$ 
with $(l\setminus\{\phi,\neg \phi\})\not\in\hypothesisSpace^\bot$.
Intuitively, non-\safe learning frameworks allow cases in which 
the target   is easy to learn if we aim at learning the \emph{whole} target,  
not a \emph{subset} of it. 
In the following,
we focus on FO learning frameworks that are
\safe\footnote{All learning from entailment results we found in the literature
could be formulated in terms of  safe learning frameworks.}. 																																								The first transferability result we present is
for the case in which the learner has access to only membership queries.
Before showing the reduction,
we define the procedure $ \algfindval$ that takes as input a precision $ p $
and a formula $ \phi $
and returns the highest valuation $\beta$ with precision $p$
of a 
formula $ \phi $ entailed by the target $ t $ (or zero if it is not entailed). 
That is, $\beta$ is such that $\beta=_p \val(\phi,t)$.
For any $ \gamma\in [0,1]_p $ the procedure can check if $ t \models (\phi,\gamma) $ 
by calling the oracle $ \MQc{\Fmf_\pi}{t} $ with 
$ (\phi,\gamma) $ as input.
To compute $\beta$  such that $\beta=_p \val(\phi,t)$, \algfindval
performs a binary search on $[0,1]_p$. Lemma~\ref{l:val} states the correctness and the complexity of  \algfindval.

\begin{lemma}\label{l:val} 
	Let $\Fmf_\pi=(\examples_\pi,\hypothesisSpace_\pi)$ be a possibilistic learning framework and 
	let	$ t\in\hypothesisSpace_\pi$ be the target.
	 \algfindval, 
		with input a precision $ p\in \mathbb{N}^+ $ and $ \phi \in\examples_\pi$, runs  in  polynomial time
	in $ p $ and $ |\phi| $
	and outputs $ \beta $ such that $ \beta =_p \val(\phi,t) $.
\end{lemma}	
\begin{proof}[Sketch]
		By Point~2 of Lemma~\ref{p:propositions}, 
	 \algfindval 
	can determine $\beta$ such that $ \beta =_p \val(\phi,t) $
	by performing a binary search on 
	the interval of numbers $[0,1]_p$.
					So the number of iterations is bounded by $ log_2 (10^{p} +1) $, which is  polynomial in $ p $.
	%Clearly,
	Each iteration can be performed in polynomial time 
	in  $ |\phi| $ and $p$.
	% (each call
	%to the membership oracle ${\sf MQ}_{\Fmf_\pi,\target}$ counts as one step of computation).
\end{proof}

%In our next theorem,
%We now show that,
By Thm.~\ref{thm:memb},
for \safe FO learning frameworks,
polynomial time results with only membership queries can be transferred 
to their possibilistic extensions if 
the precision of the target is known (by Thm.~\ref{thm:memonly}, we cannot remove this assumption).
\begin{theorem}\label{thm:memb}
	Let \Fmf be a \safe FO learning framework.
	For all $ p\in \mathbb{N^+} $, when only membership queries can be asked,
	$ \Fmf $ is in \PTimeL iff
	$ \Fmf_\pi^p $ is in \PTimeL. 
\end{theorem}
\begin{proof}
	To show  the transferability of  \PTimeL 
	membership from $ \Fmf$ to $ \Fmf_\pi $,
	we use the following claim.

	\begin{claim}\label{l:mq:findlevel}
		Assume 
		$ \Fmf  = (\examples,\hypothesisSpace)$ is \safe and in \PTimeL
		with only membership queries.
		For every $ p\in \mathbb{N^+} $ and framework
						$ \Fmf_\pi^p = (\examples_\pi,\hypothesisSpace_\pi^p) $ with
		$ t\in \hypothesisSpace_\pi^p$,
				given a valuation $ \alpha $ with $\p(\alpha)=p$, one can learn 
						$ k^*_{\overline{\alpha}} $ such that 
		$ k^*_{\overline{\alpha}} \equiv t^*_{\overline{\alpha}} $
		in  time polynomial w.r.t. $ |t| $
		with only membership queries.
	\end{claim}
	\begin{proof}
		We start the execution of a polynomial time learner $ A $ 
		for $ \Fmf $. 		Whenever $ A $ calls  \MQc{\Fmf}{t^*_{\overline{\alpha}}} with $ \phi $
		as input, we call \MQc{\Fmf_\pi}{t} with $ (\phi,\alpha + 10^{-p}) $ as input 
		and we return the same answer to $ A $.
		By Point~1 of Lemma~\ref{p:propositions},
						$ \MQc{\Fmf}{t^*_{\overline{\alpha}}}(\phi) =
		\MQc{\Fmf_\pi}{t}(\phi,\alpha + 10^{-p}) $.
								Since \Fmf is \safe, 
		 $ A $ will build a  hypothesis $ k^*_{\overline{\alpha}} $
		such that 
		$ k^*_{\overline{\alpha}} \equiv t^*_{\overline{\alpha}} $
								in polynomial time w.r.t. $ |t| $.
			\end{proof}
	
	We set $ \gamma := 0 $ and $ S := \emptyset $.
	By Claim~\ref{l:mq:findlevel} we can find in polynomial time w.r.t. $ |t| $ a 
	 hypothesis 
	$ k^*_{\overline{\gamma}} $ such that 
	$ k^*_{\overline{\gamma}}  \equiv t^*_{\overline{\gamma}}  $.
	For every $ \phi \in k^*_{\overline{\gamma}} $,
	we run  \algfindval with 
	$p= \p(t) $ and $ \phi $ as input to find $ \val(\phi,t) $.
	In this way, by Point~3 of Lemma~\ref{p:propositions} 	and Lemma~\ref{l:val}, we identify  
	in polynomial time w.r.t. $ |t| $
	some  $ \beta\in t^v\cup\{1\} $ such that $ k^*_{\overline{\gamma}} \equiv t^*_\beta $. 	We set $k^*_\beta:= k^*_{\overline{\gamma}} $ 	 and  add $ k^*_\beta $ to $ S $.
	Then, we update $ \gamma $ to the value $ \beta $ and 
	 apply  Claim~\ref{l:mq:findlevel} again.
	For every $ \phi \in k^*_{\overline{\gamma}} $, we run  \algfindval again with 
	$p= \p(t) $ and $ \phi $ as input to find $ \val(\phi,t) $.
	We repeat this process until we find  
	$ k^*_{\overline{\gamma}}\equiv \emptyset  $ or $ \gamma + 10^{-p} > 1 $. 	Each time we run  \algfindval, we identify a higher 
	valuation in $ t^v $. Therefore, this happens at most $ |t^v| $ times.	 
					For all $ \alpha\in t^v $,
	there is  $ k^*_{\alpha}\in S$ that satisfies
	$ k^*_{\alpha} \equiv t^*_{\alpha} $,
	therefore, by Lemma~\ref{l:fundamentalidea},
		\[ h = \bigcup_{k^*_{\alpha}\in S} \{ (\phi,\alpha) \mid \phi\in k^*_{\alpha} \}  \]
	is such that $ h\equiv t $.

	We now show the transferability of  \PTimeL 
	membership from $ \Fmf_\pi $ to $ \Fmf $.
	Let $ k\in\hypothesisSpace $ be the target.
	We start the execution of a learner $ A_\pi $  for $ \Fmf_\pi $
	that attempts to learn a hypothesis equivalent to
	$ t = \{ (\phi,1 ) \mid \phi\in k \} $.
	By Claim~\ref{l:target2poss} of Theorem~\ref{t:posstoclassical},
	we can simulate a
	call to \MQc{\Fmf_\pi}{t} with input $ (\phi,1) $ 
	 by calling \MQc{\Fmf}{k} with $ \phi $
	as input and returning the same answer to $ A_\pi $.
	$ A_\pi $ terminates in polynomial time w.r.t. $ |t| $
	with a hypothesis $ h $ such that $ h\equiv t $.
	As $ h^* \equiv t^* = k $, $h^*$ is as required.
\end{proof}

When we 
want to transfer learnability results
from $ \Fmf$ to $ \Fmf_\pi$ it is important to learn one $h_\alpha$ such that $h_\alpha\equiv t_\alpha$
for each $\alpha\in t^v$, where $t$ is the target (Example~\ref{ex:hypothesis}).

\begin{example}\label{ex:hypothesis}
Let $t=\{ (p\rightarrow q_1,0.3), (p\rightarrow q_2,0.7) \}$.
We can use the polynomial time algorithm for propositional Horn~\cite{DBLP:conf/icml/FrazierP93}   to
learn
a hypothesis
$ k^* = \{ p\rightarrow (q_1 \land q_2) \} \equiv  t^* $.
However, if   $ h= \{ (\phi,\val(\phi,t)) \mid \phi\in k^* \} $
then   $ h = \{ (p\rightarrow (q_1 \land q_2),0.3) \} \not\equiv t $.
\end{example}

	A learner that has access to both
	membership and equivalence query oracle
	has a way of finding the precision of the target
	when it is unknown.
	With membership queries, we can use  \algfindval 
	to find the valuation of formulas up to a given precision.
		By Lemma~\ref{l:findh},		 we 
	 can
	 obtain useful information about the precision of the target
	with the
	 counterexamples obtained 
	after  an
	equivalence query.

	\begin{lemma}\label{l:findh}		Assume $ \Fmf_\pi = (\examples_\pi,\hypothesisSpace_\pi) $ is the possibilistic 
		extension of a
		 \safe FO learning framework
		and   
		$ t\in\hypothesisSpace_\pi $ is the target.
		Given $ p \in \mathbb{N^+} $,
		one can determine
		 that $ p < \p(t) $ or
		 compute  $ h \in\hypothesisSpace_\pi$
		such that
		$ h \equiv t $, in   polynomial time w.r.t. $ |t| $, $ p $, and 
		the largest counterexample
		seen so far.
	\end{lemma}
	\begin{proof}	
In our proof, we use the following claims. 		
\begin{claim}\label{clm:aux}
Given $ h \in\hypothesisSpace_\pi$ such that $t\models h$, one can 
construct in polynomial time  in $|h|$ some $h'\in\hypothesisSpace_\pi$
such that $t\models h'\models h$ and, for all $(\phi,\alpha)\in h'$, 
$\ t\models (\phi,\alpha)$ and $\alpha=_{\p(h')}\val(\phi,t)$. 
\end{claim}
\begin{proof}
Let $h'$ be the set of all $(\phi,\beta)$ such that 
$(\phi,\alpha)\in h$ and  \algfindval returns $\beta$ with 
$\phi$ and $\p(h)$ as input. As $t\models h$, by construction of $h'$, 
$\ t\models h'\models h$.
By Lemma~\ref{l:val}, 
 $h'$ can be constructed in polynomial time in $|h|$ and  is as required.  
\end{proof}
		
\begin{claim}\label{l:counterexample} 	Let 			$ h \in\hypothesisSpace_\pi$ be 		such that, for all $(\phi,\alpha)\in h$, $ t \models (\phi,\alpha)$ and 
	$\alpha=_{\p(h)}\val(\phi,t)$. 
			If $\EQc{\Fmf_\pi}{\target}$  	with input $ h $ returns
			$ (\phi,\alpha) $ 
	then either we know that  $  \p(t) > \p(h) $ 
	or  
	$ h^*_\beta \not\models \phi $ where $\beta=_{\p(h)}\val(\phi,t)$.
\end{claim}
\begin{proof}
									 	By Point~1 of Lemma~\ref{p:propositions}, $ h^*_\beta \models \phi $ iff $ h\models (\phi,\beta) $.
	If $ h\models (\phi,\beta) $ or $ \beta =0 $ (note: $\beta$ can be $0$ because, e.g., $0.01=_1 0$), then $ \p(\val(\phi,t)) > \p(h) $.
	By Point~3 of Lemma~\ref{p:propositions},
		$ \val(\phi,t)\in t^v\cup\{1\} $, 
	so  
	$ \p(t) > \p(h) $.
\end{proof}

By Remark~\ref{remark:membership}, 
we can assume at all times in this proof that any hypothesis constructed 
is entailed by the target (possibilistic or not). Moreover, by Claim~\ref{clm:aux}, we can  assume that, for any target and hypothesis
 $t,h\in\hypothesisSpace_\pi $, we have that, 
for all $(\phi,\alpha)\in h$, $ t \models (\phi,\alpha)$ and 
	$\alpha=_{\p(h)}\val(\phi,t)$. 
	So we can assume at all times in our proof that the hypothesis $h$ we 
	construct (Equation~\ref{eq:hyp}) satisfies the conditions of Claim~\ref{l:counterexample}.

		Let $ A $  be a polynomial time learner\footnote{Assume w.l.o.g. that 		$A$ always eventually asks an equivalence query 
		until it finds an equivalent hypothesis (but may   execute 
		other steps and ask membership queries  between each equivalence query). 
						} for \Fmf. 		As in the proof of Theorem~\ref{thm:memb}, we run multiple instances of $A$. 
		We denote  by $\Rsf$ the set of instances of $A$.
		Each instance in $\Rsf$ is denoted $A_\beta$ and attempts to learn a hypothesis equivalent to 
		$t^*_\beta$,
				where $\beta$ is a valuation. 
		We sometimes write $A^n_\beta$ to indicate that the instance $A_\beta$ has asked 
		$n$ equivalence queries so far. 
		We denote by $ k^{\beta,n} $ the 		 hypothesis given as input by
		 $A^n_\beta$ when  it  asks its $n$-th equivalence query.
		 For $n=0$, we assume that $ k^{\beta,n}=\emptyset$.
		
		Initially, $\Rsf:=\{A^0_{10^{-p}}\}$.
		  Whenever $ A_\beta\in\Rsf $ asks a membership query with input 
		$ \phi\in\examples $, 
		by Point~1 of Lemma~\ref{p:propositions},
				we can simulate $ \MQc{\Fmf}{t^*_\beta} $ by calling $ \MQc{\Fmf_\pi}{t} $ with  $ (\phi,\beta) $ as input
		and  returning the same answer to $ A_\beta $. 						Let $h_0$ be $\{ (\phi_{\top},\alpha)\}$ 
where $\phi_{\top}$ is a tautology and $\alpha$ is a valuation with $\p(\alpha)=p$.	
		Whenever $ A^n_\beta \in\Rsf  $ asks its $n$-th equivalence query, 
		we leave  $A^n_\beta $  waiting in the query state (see 
		description of a learning system in Subsection~\ref{subsec:learnability}).
		When all $A^m_{\alpha}\in\Rsf$ are waiting in the query state,
		we create  
		\begin{equation} \label{eq:hyp}
		h := \bigcup_{A^m_{\alpha}\in\Rsf} \{ (\phi,\alpha) \mid \phi\in k^{\alpha,m}\}\cup h_0
				\end{equation}
		and  call $ \EQc{\Fmf_\pi}{t} $ with  $ h $ as input (note: each instance $A_{\alpha}\in\Rsf$ may have asked a different number of
		 equivalence queries when $A^n_\beta$ asks its $n$-th equivalence query).
														If the answer is `yes', we have computed $h$ such that $h\equiv t$ and we are done.
						Upon receiving a (positive) counterexample $ (\phi,\gamma) $,
		we  run  \algfindval with $\phi$ and $\p(h)$ as input and compute a valuation $\beta$ such that
		$ \beta =_{\p(h)} \val(\phi,t) $ (Lemma~\ref{l:val}). 
		If $A_\beta \not\in \Rsf$,
		  we start the execution of the instance 
		  $A_\beta$ of algorithm $A$ and add $A_\beta$ to  $\Rsf$.
		Otherwise, $A_\beta \in \Rsf$ and
		we check whether $k^{\beta,m}\models \phi$ (assume $m$ is the number of equivalence queries posed so far by $A_\beta$).
						If $k^{\beta,m}\models \phi$ then,
								by  Claim~\ref{l:counterexample}, 
 we know
		that 
		$\p(h) < \p(t) $ then we are done.
		If $k^{\beta,m}\not\models \phi$ then  $\phi$ is a (positive) counterexample for $k^{\beta,m}$ and $t^*_\beta$.
		We return  $ \phi $
		to every 
		$ A^m_{\alpha}\in\Rsf $ such that $ \alpha \leq \beta $ and
		$ k^{\alpha,m}\not\models \phi $ and these instances resume their executions.
		Observe that,		since $h_0\subseteq h$, by the construction  
		of $h$, at all times 		$ \p(h) = p $.

		We now argue that this procedure terminates in polynomial time  
		w.r.t. $ |t| $, $ p $, and 
		the largest counterexample
		seen so far. Since there is only one instance $A_\beta$ in $\Rsf$ 		for each valuation  $\beta$ such that
		$ \beta =_p \val(\phi,t) $, by Point~3 of Lemma~\ref{p:propositions}, we have that
		at all times $|\Rsf|$ is linear in $|t^v|$, which is bounded by $|t|$.
		By Lemma~\ref{l:val}, 
		whenever we run  \algfindval to compute a valuation with $\phi$ and $p$ as input, 
		only polynomially many steps in $|\phi|$ and $p$ are needed. 
		Since \Fmf is \safe and $A$ is a polynomial time learner for \Fmf either we can determine that $p<\p(t)$ or  each  $A_\beta\in\Rsf$
		terminates, in polynomial time in the size of $t^*_\beta$ and the 
		largest counterexample seen so far, and outputs $k^{\beta,n}=h^*_\beta$  
		such that 
		$h^*_\beta\equiv t^*_\beta$. In this case, 		 by Lemma~\ref{l:fundamentalidea},
		  $h\equiv t$ and the process terminates. 	
												  																																																																																																																																													\end{proof}

			\begin{figure}[]		\centering
		\begin{tikzpicture}
		\tikzset{d/.style={circle,fill=black,inner sep=0pt,minimum size=4pt}}
		\tikzset{l/.style={fill=white,opacity=0,text opacity=1}}
		\tikzset{h/.style={ultra thin,-,dashed,opacity=0.3}}
		
		\def\x{1}
		\def\y{0}
		
		\node[l] () at (\x+6.5,\y-.2) {$ {\bf (a)}$};
		\node[l] () at (\x+6.5,\y-.8) {$ {\bf (b)}$};
		\node[l] () at (\x+6.5,\y-1.7) {$ {\bf (c)}$};
		
				\node[l] () at (\x+4,\y-.2) 
		{\small{$ \EQc{\Fmf_\pi}{t}(h) = (p \rightarrow q_1,0.1)$}};
		\node[l] () at (\x+4,\y-.8) 
		{\small{$ \EQc{\Fmf_\pi}{t}(h) = (p \rightarrow q_1,0.1)$}};
		\node[l] () at (\x+4.12,\y-1.5) 
		{\small{$ \EQc{\Fmf_\pi}{t}(h') = (p \rightarrow q_2,0.21)$}};
				\def\x{0}
		\def\y{0}
						\node[l] () at (\x+2.2,\y-1) {{\tiny$ \Leftarrow p \rightarrow q_1$}};
		\node[l] () at (\x+2.2,\y-1.2) {{\tiny$ \Leftarrow p \rightarrow q_1$}};
						
		\node[l] () at (\x,\y+0.2) {$ A_{0.1}$};
		\node[l] () at (\x+1,\y-0.4) {$ A_{0.3}$};
		
		\node[l] () at (\x+2,\y-1.7) {$ A_{0.7}$};
		
				\draw[-] (\x,\y) to  (\x,\y-0.2);
		
		\draw[-,dotted] (\x,\y-0.2) to  (\x,\y-0.8);
		\draw[-] (\x+1,\y-0.6) to  (\x+1,\y-.8);
		
		\draw[-,dotted] (\x,\y-0.8) to  (\x,\y-1.2);
		\draw[-,dotted] (\x+1,\y-0.8) to  (\x+1,\y-1);
		
		\draw[-] (\x+1,\y-1) to  (\x+1,\y-1.3);
		\draw[-] (\x,\y-1.2) to  (\x,\y-1.5);
		\draw[-,dotted] (\x+1,\y-1.3) to  (\x+1,\y-2);
		\draw[-,dotted] (\x,\y-1.5) to  (\x,\y-2);
		\draw[-] (\x+2,\y-1.9) to  (\x+2,\y-2);
		
		\def\z{2.9}
		\def\zz{1.5}
		
				\draw[h] (\x,\y-0.21) to  (\z,\y-0.21);
		\draw[h] (\x,\y-0.81) to  (\z,\y-0.81);
		\draw[h] (\x,\y-1) to  (\zz,\y-1);
		\draw[h] (\x,\y-1.2) to  (\zz,\y-1.2);
		\draw[h] (\x,\y-1.5) to  (\z,\y-1.5);

		\end{tikzpicture}
		\caption{
			Multiple instances of   algorithm
			$A$ in Example~\ref{e:mqeq}. 			Time flows  top-down.
						A dotted line means that the learner is waiting in query state,
			a continuous line means that the learner is running. 
											}
		\label{f:mqeq}
	\end{figure}
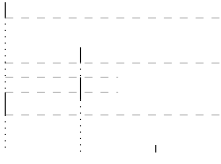
	
	The constructive proof of Lemma~\ref{l:findh}
	delineates the steps made in
	 Example~\ref{e:mqeq}
	 where the precision of the target is 1.

	\begin{example}\label{e:mqeq}\upshape
		Let $ \Fmf = (\examples,\hypothesisSpace) $ be 
		the \safe
		learning framework where $ \hypothesisSpace$ is the set of all propositional 
		Horn KBs and $ \examples $ is the set of all (propositional) Horn clauses. 						Let $ t\in\hypothesisSpace_\pi $ and $A$ be, respectively, the target and the learner of  Example~\ref{ex:hypothesis}. 
				Following our argument in Lemma~\ref{l:findh}, we start  an instance  $ A_{0.1} $ of $A$.
		When $ A_{0.1} $ is waiting in the query state,
		we build $ h = \{ (\phi_{\top},0.1) \} $ (Equation~\ref{eq:hyp}) and call $ \EQc{\Fmf_\pi}{t} $ 
		with  $h$ as input 
		(Point~$ (a) $ in Figure~\ref{f:mqeq}).
		Assume we receive the positive counterexample $ (p\rightarrow q_1,0.1) $.
		We run  \algfindval with   $ 1 $ and $ p\rightarrow q_1 $ as input, which computes
		 $ \val(p\rightarrow q_1,t) = 0.3 $.
		Since $ A_{0.3}\not\in\Rsf$, we start  $ A_{0.3} $.
				When all learners 
		are waiting in the query state, 		we call again $ \EQc{\Fmf_\pi}{t} $ 
		with  $ h $ as input
		(Point~$ (b) $ in Figure~\ref{f:mqeq}).
		At this point, $\Rsf=\{A_{0.1},A_{0.3}\}$.
		
		Assume we receive  $ (p\rightarrow q_1,0.1) $ again. 
						We have that $ \val(p\rightarrow q_1,t) = 0.3 $ and $ A_{0.3}\in\Rsf$.
		Since  $k^{0.3,1}\not\models p\rightarrow q_1$ and 
		$k^{0.1,1}\not\models p\rightarrow q_1$, 		we return $p\rightarrow q_1$ to both $ A^1_{0.1} $ and $ A^1_{0.3} $ and they resume their executions.
				All learners will eventually 
								be waiting in query state. When this happens
		we call 
		$ \EQc{\Fmf_\pi}{t} $ 
		with  $ h' = \{ (\phi_{\top},0.1), 
								(p \rightarrow q_1,0.1),
								(p \rightarrow q_1,0.3) \} $ as input.
		
		Assume the response is  $ (p\rightarrow q_2,0.21) $. We run  \algfindval
		with  $ 1 $ and $ p\rightarrow q_2 $ as input, which returns 		$ \val(p\rightarrow q_2,t)  = 0.7$.
		As before, we start $ A_{0.7} $ (Point~$ (c) $ in Figure~\ref{f:mqeq}) and add it to $\Rsf$.
		When all learners are 
		waiting again 		we call $ \EQc{\Fmf_\pi}{t} $ 
		with  $ h' $ as input.
		Assume we receive  $ (p\rightarrow q_2,0.1) $. We then 
						  send $ p\rightarrow q_2 $
		to every learner in $\Rsf$.
		Next time we call $ \EQc{\Fmf_\pi}{t} $,
		with $ h' \cup 
								\{ (p\rightarrow q_2,0.7),
								(p\rightarrow q_2,0.3),
								(p\rightarrow q_2,0.1) \} $ as input.
		The answer is `yes' and we are done. \hfill {\mbox{$\triangleleft$}}
	\end{example}

	%~ In some cases, the learner can discover if the precision of the hypothesis needs to increase
		%~ (Example~\ref{e:lowprec}).
			
	%~ \begin{example}\label{e:lowprec}\upshape
		%~ Assume the target is $ t = \{ (p \rightarrow q,0.32) \} $	
		%~ and the learner built the hypothesis $ h=\{ (p \rightarrow q,0.3) \} $. 
		%~ Similarly 		to Example~\ref{e:mqeq}, the precision of the hypothesis is set to $1$.
		%~ A future equivalence query will return the counterexample  $ h=\{ (p \rightarrow q,\alpha) \} $
		%~ with $ \alpha > 0.3 $. The learner will run  \algfindval with input $ 1 $ and $ p \rightarrow q $, which will return $ 0.3 $. Since $ h\models  (p \rightarrow q,0.3)$,
		%~ this can happen only if the precision of the hypothesis is low. \hfill {\mbox{$\triangleleft$}}
	%~ \end{example}

A direct consequence of Lemma~\ref{l:findh} is Theorem~\ref{thm:mqeq}.\	
	
\begin{theorem}\label{thm:mqeq}
	For every \safe  FO learning frameworks $\Fmf $
	we have,
			$ \Fmf $ is in \PTimeL
	iff 
	$ \Fmf_\pi  $ is in \PTimeL. 
\end{theorem}
\begin{proof}
	One direction holds by Theorem~\ref{t:posstoclassical}.
	We prove the other direction.
	Let $\Fmf $ be a \safe  FO learning framework in \PTimeL and let 
	$ \Fmf_\pi=(\examples_\pi,\hypothesisSpace_\pi) $ be its possibilistic extension. 
	Consider a learner that initially estimates 
		precision $p$ of the target $t\in\hypothesisSpace_\pi$ to be $1$.
	Using Lemma~\ref{l:findh}, we can assume that this learner can either determine that    $ p < \p(t) $ 
	or  find a 
	hypothesis $ h $ such that $ h\equiv t $, 
	in  time polynomial with respect to $ |t| $, $ p $ and 
		the largest counterexample
		seen so far.
	In the former case,
	this learner sets the estimated precision $ p $ of the target to $ p+1 $.
				This happens at most $ \p(t) $ times, which is bounded by $ |t| $. 
	As a consequence,
	$ \Fmf_\pi  $ is  in \PTimeL.	
\end{proof}	

We end this section recalling that our results can be transferred
to the PAC model~\cite{Valiant} extended with membership queries
 %~ a connection between the exact and the PAC learning models. 
%~ In the PAC model, a learner receives classified
%~ examples drawn according to a probability distribution
%~ and attempts to create a hypothesis that approximates
%~ the target. 
%~ It is known that polynomial time results for the exact learning model 
%~ can be transferred to the PAC learning model~\cite{Valiant} extended with membership 
%~ queries
(Theorem~\ref{th:pacmq}).
 \begin{theorem}[\cite{angluinqueries,MohEtAl}]\label{th:pacmq}
Let \PTimePL be the class of all learning frameworks that are 
 PAC learnable with membership queries in polynomial time. Then, 
$\PTimeL \subseteq \PTimePL$.
\end{theorem}
 	
By Theorems~\ref{thm:mqeq} and~\ref{th:pacmq}, the following holds. 
\begin{corollary}
For all \safe FO learning frameworks \Fmf, 
if  $\Fmf\in\PTimeL$ then $\Fmf_\pi\in\PTimePL$.
\end{corollary}

      \section{Conclusion}

Uncertainty is  widespread  in
learning processes. Among different   uncertainty formalisms,  possibilistic logic 
stands out because of its ability to express  preferences among worlds and model ignorance.
We presented the first study on the exact
(polynomial) learnability of possibilistic theories. 
It follows from our results  that   various algorithms 
designed for exact learning fragments of first-order logic 
can be adapted to learn their possibilistic extensions. 
  We leave  open the problem of polynomial time transferability with only equivalence queries.

\section*{Acknowledgements}
We are supported by the University of Bergen. We would like to thank Andrea Mazzullo
for joining initial discussions. 

\bibliographystyle{named}
\bibliography{biblio}

\end{document}